\documentclass[Journal]{IEEEtran}
\usepackage{amsbsy}
\usepackage{floatflt} 
\usepackage{xcolor}
\usepackage{url}
\usepackage{amsmath}
\usepackage{amssymb}
\usepackage{times}
\usepackage{graphicx}
\usepackage{xspace}
\usepackage{paralist} 
\usepackage{setspace} 
\usepackage{xypic}
\xyoption{curve}
\usepackage{latexsym}
\usepackage{theorem}
\usepackage{ifthen}
\usepackage{subfigure}
\usepackage{turnstile}
\usepackage{cite}
\usepackage{booktabs}

{\theoremheaderfont{\it} \theorembodyfont{\rmfamily}
\newtheorem{theorem}{Theorem}

\newtheorem{corollary}[theorem]{Corollary}

}
\begin{document}
\title{Distributed Bayesian Quickest Change Detection in Sensor Networks via Two-layer Large Deviation Analysis}
\author{
        Di~Li,~\IEEEmembership{Student Member,~IEEE,}
       Soummya~Kar,~\IEEEmembership{Member,~IEEE,}
       Fuad E. Alsaadi, \\
        Shuguang Cui,~\IEEEmembership{Fellow,~IEEE}
        \thanks{D. Li and S. Cui are with the Department of Electrical and Computer Engineering, Texas A\&M University, College Station, TX 77843 USA (e-mail: dili@tamu.edu; cui@tamu.edu).} 
        \thanks{S. Kar is with the Department of Electrical and Computer Engineering, Carnegie Mellon University, Pittsburgh, PA 15213 USA (e-mail: soummyak@andrew.cmu.edu).}
        \thanks{F. E. Alsaadi is with the Department of Electrical and Computer Engineering, King Abdulaziz University, Jeddah, 22254 Saudi Arabia (e-mail: fuad\_alsaadi@yahoo.com).}
%
}

\maketitle

\begin{abstract}
We propose a distributed Bayesian quickest change detection algorithm for sensor networks, based on a random gossip inter-sensor communication structure. Without a control or fusion center, each sensor executes its local change detection procedure in a parallel and distributed fashion, interacting with its neighbor sensors via random inter-sensor communications to propagate information. By modeling the information propagation dynamics in the network as a Markov process, two-layer large deviation analysis is presented to analyze the performance of the proposed algorithm. The first-layer analysis shows that the relation between the probability of false alarm and the conditional averaged detection delay satisfies the large deviation principle, implying that the probability of false alarm according to a rare event decays to zero at an exponentially fast rate when the conditional averaged detection decay increases, where the Kullback-Leibler information number is established as a crucial factor. The second-layer analysis shows that the probability of the rare event that not all observations are available at a sensor decays to zero at an exponentially fast rate when the averaged number of communications increases, where the large deviation upper and lower bounds for this rate are also derived, based on which we show that the performance of the distributed algorithm converges exponentially fast to that of the centralized one, by proving that the defined distributed Kullback-Leibler information number converges to the centralized Kullback-Leibler information number.
\end{abstract}
\begin{keywords}
Quickest change detection, distributed detection, large deviation, sensor networks, random gossip, Bayesian model, Markov process.
\end{keywords}

\section{Introduction}
\IEEEPARstart{Q}uickest change detection problems focus on detecting abrupt changes in stochastic processes as quickly as possible, with constraints to limit the detection error. Quickest change detection has wide applications in fields such as signal and image processing \cite{Lai-QD-cognitive,Husheng-Quickest,Trivedi-secret}, computer network intrusion detection \cite{Thottan-Anomaly,Tar-intrusion,Cardenas-Mac},
neuroscience \cite{Commenges-neuro}, environment and public health surveillance \cite{Frisen-public,Sonesson-public}, and system failure detection \cite{Rice-structure,Mainwaring}. Specifically, when quickest change detection is implemented in sensor networks \cite{VVV-Energy Efficient,Mei,Tartakovsky-quickest03}, it can detect the change of statistical features, such as the mean and variance, over the observation sequences taken by sensors. For example, quickest change detection can be implemented in sensor networks for chemical industry to monitor the leakage, or to surveille the change of temperature in the field, by detecting the change in statistical patterns.

For signal processing implementation in sensor networks, essentially it can be divided into the following two categories: centralized vs. distributed algorithms. For centralized quickest change detection algorithms \cite{Tartakovsky08asymptoticallyoptimal,VVV-decentr2001,Poor-oneshot,Mous-decentralized,LZ-decentra,Di-GlobalSIP,Ban-Efficiency-quickest}, a control or fusion center exists to process the data in a centralized way. Specifically, in centralized algorithms, they assume that either the raw observations from all the sensors or certain pre-processed information from the sensors (some people call this case as decentralized sensing) are available to the control or fusion center via certain communication channels; then a final centralized detection procedure is executed at the center. However, centralized algorithms have some disadvantages, such as heavy communication burden, high computation complexity, low scalability, and poor robustness. On the contrary, distributed implements do not require a control or fusion center, and the detection procedure is implemented at each sensor in a local and parallel fashion, with interactions among sensors in the neighborhood to exchange information. While centralized quickest change detection algorithms have been well-studied, there are fewer literatures on the study of distributed algorithms for quickest change detection problems \cite{Braca-distributedchangeconsensus,stankovic-distributedchange}, which become more desired in large-scale networks with a huge volume of data, in order to reduce the overall computation complexity and to enhance scalability. In \cite{Braca-distributedchangeconsensus}, a distributed consensus based Page's test algorithm, using cumulative sum (CUSUM) log-likelihood of the data, was proposed, with the assumption that the change happening time is deterministic but unknown, which is called a non-Bayesian setup. In \cite{stankovic-distributedchange}, a distributed change detection algorithm was proposed, to combine a global consensus scheme with the geometric moving average control charts that generate local statistics.

In both \cite{Braca-distributedchangeconsensus} and \cite{stankovic-distributedchange}, non-Bayesian setups of the change happening time are considered, where the communication stage and the observation stage are interleaved, i.e., they are at the same time scale and each is executed once within one system time slot. Under such an interleaving strategy, the convergence of the test statistic is established when the system time goes to infinity. However, this type of convergence analysis over time does not fit well into quickest change detection problems, which are time-sensitive, with the goal to detect the change as quickly as possible. This is different from traditional detection problems without much consideration of the timing issue, where the convergence analysis is commonly performed as the system time goes to infinity.

Different from the existing work, in this paper we propose a distributed change detection algorithm based on a Bayesian setup of change happening time. To the best of our knowledge, this paper is the first work discussing the distributed change detection algorithm under such a Bayesian setup. Additionally, in our proposed distributed algorithm, multiple communication steps are in between two observation instants, i.e., the communication step has a smaller time scale than that of the observation stage. In communication steps, a random point-to-point gossip based algorithm is proposed as in \cite{Di-LDKalman,Di-KalmanQuanize}. We model the information propagation procedure governed by this communication procedure as a Makov process. We then analyze the performance of the proposed distributed change detection algorithm, with a method of two-layer large deviation analysis. Large deviation techniques \cite{DEMBO-LD,Bucklew-book} have been used to analyze the performance of either centralized or distributed estimation and detection algorithms, for example, in \cite{Di-LDKalman,Baj-LD,Jak-Detection-LD,Sahu-SPRT-TSP-15}. However, no existing work has utilized the technique of large deviation analysis to study the performance of the change detection algorithms, especially the distributed change detection algorithms. The most related work is \cite{Sahu-SPRT-TSP-15}, in which a
distributed sequential detection method is proposed to solve the problem of Gaussian binary hypothesis testing. The sequential hypothesis testing problem could be considered as a special case of change detection problems, where the change happened at the initial time point \cite{Poor}.

In this paper, the first-layer large deviation analysis shows that the relation between the conditional averaged detection delay and the probability of false alarm satisfies the large deviation principle, which implies that the probability of false alarm decays exponentially fast as the conditional averaged detection delay increases. In the first-layer analysis, the nonlinear renewal theorem is adopted, by representing the stopping time with the form of a random walk crossing a constant threshold plus a nonlinear term. The second-layer analysis derives the large deviation upper and lower bounds for the probability of the rare event that not all observations are available at a sensor. Based on this, we further prove that the distributed Kullback-Leibler information number converges to the centralized Kullback-Leibler information number, by deriving the upper and lower bounds for the distributed form of Kullback-Leibler information numbers. We eventually show that the performance of the distributed algorithm converges exponentially fast to that of the centralized one when the averaged number of communications increases. In the analysis, the concept of hitting time in Markov chain is used to derive the large deviation upper and lower bounds.

The rest of the paper is organized as follows. Section~\ref{System_Setup} sets up the system model and describes the quickest change detection problem. Section~\ref{Large_Deviation_Analysis} presents the large deviation analysis in the centralized change detection algorithm to set up the background. Section~\ref{Distributed_Change_Detection} introduces the distributed change detection algorithm, and develops the corresponding two-layer large deviation analysis. Section~\ref{simulation} provides the simulation results to validate the analytical results from the previous sections.  Section~\ref{conclusion} concludes the paper.

\section{System Setup}\label{System_Setup}
Consider a network with $N$ nodes. Assume that a change happens at time $\lambda=k$. Then conditioned on $\lambda=k$, independent and identically distributed (i.i.d.) observations $X^i_1,\cdots,X^i_{k-1}$ at sensor $i$ follow a distribution with density function $f^i_0(x)$; observations $X^i_k,X^i_{k+1}\cdots$ follow another distribution with density function $f^i_1(x)$. We assume that observations at different sensors are independent of each other and the various densities are absolutely continuous with respect to the Lebesgue measure. Denote $\mathbf{X}^{i}_{n}=[X^i_1,\cdots,X^i_n]$ as observations up to time $n$ at node $i$. Let $\mathbb{P}_k$ be the probability measure of $\mathbf{X}^{i}_{n}$ when the change occurs at time $k$, and $\mathbb{E}_k$ be the corresponding expectation operator. We need to design a sequential on-line detection algorithm (with a stopping criterion) over the observation sequence to detect the change.

Consider a Bayesian setup, and assume the prior distribution for the change-point time $\lambda$ as
\[\pi_k=\mathbb{P}(\lambda=k).\]
Let $\mathbb{P}^\pi$ denote the probability measure, defined as $\mathbb{P}^\pi(\cdot)=\sum_{k=1}^{\infty} \pi_k \mathbb{P}_k(\cdot)$, and let $\mathbb{E}^\pi$ denote the expectation operator with respect to the measure $\mathbb{P}^\pi$.

The change detection problem can be converted to the hypothesis testing problem with hypotheses $``H_0: \lambda > n"$ and $``H_1: \lambda \leq n"$, i.e., to sequentially decide which hypothesis is true at each time $n$. If $H_0$ is decided, it indicates that the change hasn't happened; if $H_1$ is decided, it claims that the change has happened.
\subsection{Centralized Scheme}
First we discuss the centralized change detection algorithm, which means that observations from all sensors are available at a control center, where the detection algorithm is performed. Denote $\mathbf{X}_n=[\mathbf{X}^{1}_{n},\cdots,\mathbf{X}^{N}_{n}]$ as observations up to time $n$ from all sensors; denote the likelihood ratio for $``H_1: \lambda \leq n"$ vs. $``H_0: \lambda>n"$ averaged over the change point (see \cite{Tart-Bayesian}) as:
\begin{align}
\Lambda_n&=\frac{\mathbb{P}(\mathbf{X}_n|\lambda\leq n)\mathbb{P}(\lambda\leq n)}{\mathbb{P}(\mathbf{X}_n|\lambda> n)\mathbb{P}(\lambda> n)}\nonumber\\
&=\frac{\sum_{k=1}^n \left[\pi_k \prod_{i=1}^N \prod_{j=k}^n f^i_1(X^i_j) \prod_{j=1}^{k-1} f^i_0(X^i_j)\right]}{\sum_{k=n+1}^\infty \pi_k \prod_{i=1}^N \prod_{j=1}^n f^i_0(X^i_j)   }.
\end{align}

Assume the prior distribution is geometric \cite{Poor}, i.e.,
\[\pi_k=\rho (1-\rho)^{k-1},~\mbox{with}~\rho~\mbox{in}~(0,1).\]

Then, we have
\begin{equation}
\Lambda_n=\frac{1}{(1-\rho)^n}{\sum_{k=1}^{n}\pi_k \prod_{j=k}^{n}\prod_{i=1}^N \frac{f^i_1(X^i_j)}{f^i_0(X^i_j)}}.
\end{equation}
We further have the following recursive form as
\begin{equation}\label{test_stat}
\Lambda_n=\frac{1}{1-\rho}(\Lambda_{n-1}+\rho)\prod_{i=1}^{N}\frac{f^i_1(X^i_n)}{f^i_0(X^i_n)},
\end{equation}
with the initial state $\Lambda_0=0$. Taking logarithms on both sides, we have
\begin{equation}\label{rc_centr}
\log \Lambda_n=\log \frac{1}{1-\rho}+\log(\Lambda_{n-1}+\rho)+\sum_{i=1}^{N}\log\frac{f^i_1(X^i_n)}{f^i_0(X^i_n)}.
\end{equation}

Let $\mathcal{F}_n^X = \sigma (\mathbf{X}_n)$ be the $\sigma-$algebra generated by the observations $\mathbf{X}_n$, and we denote
\begin{equation}\label{posterior_prob}
p_n=\mathbb{P}\{\lambda \leq n | \mathcal{F}_n^X\}
\end{equation}
as the posterior probability that the change has occurred before time $n$. It follows that $\Lambda_n=p_n/(1-p_n)$.

We intend to detect the change as soon as possible, with a constraint on the detection error. Thus, the change detection problem can be formulated as the following optimization problem over certain decision rules:
\begin{align}\label{op_problem}
&\inf_{\tau\in \Delta(\alpha)} \mbox{ADD} (\tau) \nonumber\\
&\mbox{s.~t.}~\Delta(\alpha)=\{\tau:\mbox{PFA}(\tau)\leq \alpha\},
\end{align}
where the Averaged Detection Delay (ADD) is
\[\mbox{ADD}(\tau)=\mathbb{E}^\pi(\tau-\lambda|\tau\geq \lambda),\]
the Probability of False Alarm (PFA) is
\[\mbox{PFA}(\tau)=\mathbb{P}^\pi(\tau<\lambda)=\sum_{k=1}^\infty \pi_k \mathbb{P}_k(\tau<k),\]
with $\mathbb{E}^\pi$ and $\mathbb{P}^\pi$ defined at the beginning of this section, and $\alpha$ the upper limit of PFA.

The optimal solution to this problem is given by the Shiryaev test (see \cite{Shiryaev-1963,Shiryaev-1978}), where the detection strategy corresponds to claiming a change when the likelihood ratio $\Lambda_n$ exceeds a threshold, i.e., the optimal stopping time $\tau^*$ is
\begin{equation}\label{opt_test}
\tau^*(A)=\inf\{n\geq 1: \Lambda_n\geq A\},
\end{equation}
where $A$ is chosen such that $\mbox{PFA}(\tau^*(A))=\alpha$. It is difficult to set a threshold $A$ exactly matching the above condition. We could set $A=(1-\alpha)/\alpha$ guaranteeing that $\mbox{PFA}(\tau^*(A)) \leq \alpha$, which is due to the fact that $\mathbb{P}^\pi(\tau^*(A)<\lambda) = \mathbb{E}^\pi \left(1-p_{\tau^*(A)}\right)$ and $1-p_{\tau^*(A)} \leq 1/(1+A)$ with $p_{\tau^*(A)}$ defined in \eqref{posterior_prob}, such that $\mbox{PFA}(\tau^*(A)) \leq 1/(1+A)$. Therefore, setting $A=(1-\alpha)/\alpha$ guarantees $\mbox{PFA}(\tau^*(A)) \leq \alpha$ \cite{Tart-Bayesian}.

\subsection{Isolated Scheme}
If there is no control center and each sensor implements the local change detection algorithm purely based on its own observations, the log-likelihood ratio for hypotheses $``H_0: \lambda \leq n"$ vs. $``H_1: \lambda>n"$ of sensor $i$ at time $n$ is derived as
\begin{equation}\label{rc_isolated}
\log \Lambda^i_n=\log \frac{1}{1-\rho}+\log(\Lambda^i_{n-1}+\rho)+\log\frac{f^i_1(X^i_n)}{f^i_0(X^i_n)},
\end{equation}
with the initial state $\Lambda^i_0=0$.

Then, to solve the optimization problem in \eqref{op_problem} at sensor $i$, the Shiryaev test with test statistic in \eqref{rc_isolated} is the optimal solution \cite{Shiryaev-1963,Shiryaev-1978}, with the optimal stopping time ${\tau^i}^*$ at sensor $i$ as
\begin{equation}\label{opt_test_is}
{\tau^i}^*(A)=\inf\{n\geq 1: \Lambda^i_n\geq A\},
\end{equation}
where $A$ is chosen such that $\mbox{PFA}({\tau^i}^*(A))=\alpha$. Since this detection strategy is exclusively based on local observations at each sensor, it is called the isolated scheme.

Intuitively, the larger the difference between densities $f^i_1(x)$ and $f^i_0(x)$ is, the faster the change can be detected. To quantify the difference between densities $f^i_1(x)$ and $f^i_0(x)$, the Kullback-Leibler information number is defined as
\begin{equation}\label{KL_is}
D(f^i_1,f^i_0)=\int \log \left\{\frac{f^i_1(x)}{f^i_0(x)}\right\} f^i_1(x) dx,
\end{equation}
 which is also called divergence or KL distance between densities $f^i_1(x)$ and $f^i_0(x)$. We assume a mild condition that $0<D(f^i_1,f^i_0)<\infty$ and $0<D(f^i_0,f^i_1)<\infty$, for each $i$.

In the sequel, we will show that the Kullback-Leibler information number is a crucial factor in analyzing the performance of the change detection algorithms.

\section{Large Deviation Analysis for Centralized and Isolated Algorithms}\label{Large_Deviation_Analysis}

Large deviation studies the asymptotic behavior of a rare event. Generally, for a rare event satisfying the large deviation principle, the probability of this rare event occurring decays to zero at an exponentially fast rate in the asymptotic sense over certain quantity. In this section, we analyze the performance of the centralized algorithm, by quantifying the relation between the conditional ADD and the PFA via large deviation analysis, showing that the event of false alarm can be considered as a rare event and the corresponding PFA decays to zero exponentially fast, when the conditional ADD increases. The results in this section will set the background for analyzing the distributed case in the next section.

Since ADD might be difficult to characterize, following \cite{Tart-Bayesian}, we instead analyze the conditional ADD (CADD). The CADD is defined as $\mbox{CADD}_{k}(\tau)=\mathbb{E}_{k}(\tau-k|\tau\geq k),~k=1,2,\cdots$. The relation between ADD and CADD is described as follows:
\begin{align}
\mbox{ADD}(\tau)&=\mathbb{E}^\pi(\tau-\lambda|\tau\geq \lambda)\nonumber\\
&=\frac{\sum_{k=1}^{\infty} \pi_k \mathbb{P}_k(\tau \geq k) \mathbb{E}_k(\tau-k|\tau\geq k)}{\mathbb{P}^\pi\{\tau\geq\lambda\}} \nonumber\\
&=\frac{\sum_{k=1}^{\infty} \pi_k \mathbb{P}_k(\tau \geq k) \mbox{CADD}_{k}(\tau)  }{\mathbb{P}^\pi\{\tau\geq\lambda\}}.
\end{align}

According to the optimal stopping rule \eqref{opt_test} and the test statistic \eqref{test_stat}, we find $\mbox{CADD}_1(\tau^*)\geq \mbox{CADD}_k(\tau^*)$, for $k\geq 2$, which is explained as follows. For $k=1$ (which means that the change happens at time $1$), by investigating \eqref{test_stat}, $\Lambda_1$ is updated based on the initial state $\Lambda_0=0$. For $k\geq 2$, by investigating \eqref{test_stat}, $\Lambda_k$ is updated based on $\Lambda_{k-1}$, where $0\leq\Lambda_{k-1}< A$ according to the optimal stopping rule \eqref{opt_test} and the condition $\tau^*\geq k$. Thus, we have $\Lambda_{k-1}\geq \Lambda_0$. According to the optimal stopping rule \eqref{opt_test}, the spent time of crossing the threshold after the change happens (detection delay) in the case of $k\geq 2$ is less than that in the case of $k=1$ on average. Therefore, we have $\mbox{CADD}_1(\tau^*)\geq \mbox{CADD}_k(\tau^*)$. Additionally, the difference between $\mbox{CADD}_1(\tau^*)$ and $\mbox{CADD}_k(\tau^*)$ could be treated as a constant for large $A$, which approximately equals $\mathbb{E}_{\infty}(\log \Lambda_{k-1}), k\geq 2$ \cite{Tart-Bayesian}. Therefore, in the sequel, we focus on the use of $\mbox{CADD}_1(\tau^*)$, which could be also considered as the worst-case study.

The relation between $\mbox{CADD}_1(\tau^*)$ and $\mbox{PFA}(\tau^*)$, for the centralized scheme, is presented in the following theorem.
\begin{theorem}\label{Thm_cen}
The probability of false alarm ($\mbox{PFA}(\tau^*)$), with the optimal stopping rule \eqref{opt_test}, satisfies the large deviation principle, in the asymptotic sense with respect to the increasing conditional ADD ($\mbox{CADD}_1(\tau^*)$), i.e.,
\begin{align}
\lim_{\mbox{CADD}_1(\tau^*)\rightarrow \infty} & \frac{1}{\mbox{CADD}_1(\tau^*)} \log [\mbox{PFA}(\tau^*)]\nonumber\\
&=-( \mathcal{D}+|\log(1-\rho)|),
\end{align}
where $\mathcal{D}$ is the sum of the Kullback-Leibler information numbers across all sensors, i.e., $\mathcal{D}=\sum_{i=1}^N D (f_1^i,f_0^i)$, and $\mathcal{D}+|\log(1-\rho)|$ is the large deviation decay rate, quantifying how fast the probability of false alarm decays to zero over the increasing conditional ADD.
\end{theorem}
\begin{proof}
Recall Theorem~5 in \cite{Tart-Bayesian}, which establishes the following results:
\begin{align}
\mbox{PFA}(\tau^*)&=\frac{\zeta(\rho,\mathcal{D})}{A}(1+o(1)),~\mbox{as}~A \rightarrow \infty; \label{thm5_1}\\
\mathbb{E}_1(\tau^*)&=\frac{1}{\mathcal{D}+|\log(1-\rho)|}\left[\log \frac{A}{\rho}-\xi(\rho,\mathcal{D})\right]+o(1),\nonumber\\
&~~~~~~~~~~~~~~~~~~~~~~~~~~~\mbox{as}~A \rightarrow \infty,\label{thm5_2}
\end{align}
where $\mathcal{D}=\sum_{i=1}^N D (f_1^i,f_0^i)$, and both $\zeta(\rho,\mathcal{D})$ and $\xi(\rho,\mathcal{D})$ are functions of $\rho$ and $\mathcal{D}$. Since $\rho$ and $\mathcal{D}$ are constants once the system parameters are set, $\zeta(\rho,\mathcal{D})$ and $\xi(\rho,\mathcal{D})$ are also system constants.

Since $\mbox{CADD}_1(\tau^*)=\mathbb{E}_1(\tau^*-1)=\mathbb{E}_1(\tau^*)-1$, by combining \eqref{thm5_1} and \eqref{thm5_2}, we have
\begin{align}\label{thm1_eq1}
&\log\frac{\mbox{PFA}(\tau^*)\rho}{\zeta(\rho,\mathcal{D})(1+o(1))}=-\mbox{CADD}_1(\tau^*)({\mathcal{D}+|\log(1-\rho)|})\nonumber \\
&-\xi(\rho,\mathcal{D})+(o(1)-1)({\mathcal{D}+|\log(1-\rho)|}).
\end{align}
Then, after dividing the left-hand and right-hand sides of \eqref{thm1_eq1} by $\mbox{CADD}_1(\tau^*)$ and taking the limit as $\mbox{CADD}_1(\tau^*) \rightarrow \infty$, we have
\begin{align}
\lim_{\mbox{CADD}_1(\tau^*) \rightarrow \infty} & \frac{1}{\mbox{CADD}_1(\tau^*)} \log \mbox{PFA}(\tau^*)\nonumber\\
&= -({\mathcal{D}+|\log(1-\rho)|}).
\end{align}
\end{proof}

The above theorem quantifies the tradeoff between two performance metrics: PFA and $\mbox{CADD}_1$, in the defined change detection problems, i.e., as $\mbox{CADD}_1$ increases, PFA decays to zero exponentially fast and the decay rate is ${\mathcal{D}+|\log(1-\rho)|}$.

For the isolated scheme, at each node $i$, the relation between $\mbox{PFA}({\tau^i}^*)$ and $\mbox{CADD}_1({\tau^i}^*)$ has a similar format to that in the centralized case shown in Theorem \ref{Thm_cen}. We give the following corollary.
\begin{corollary}\label{Coro_iso}
The probability of false alarm ($\mbox{PFA}({\tau^i}^*)$), with the optimal stopping rule \eqref{opt_test_is}, satisfies the large deviation principle, in the asymptotic sense with respect to the increasing conditional ADD ($\mbox{CADD}_1({\tau^i}^*)$), i.e.,
\begin{align}
\lim_{\mbox{CADD}_1({\tau^i}^*)\rightarrow \infty} &\frac{1}{\mbox{CADD}_1({\tau^i}^*)} \log [\mbox{PFA}({\tau^i}^*)]\nonumber\\
&=-({D}(f_1^i,f_0^i)+|\log(1-\rho)|),
\end{align}
which implies that the large deviation decay rate of the PFA is ${D}(f_1^i,f_0^i)+|\log(1-\rho)|$.
\end{corollary}

Theorem~\ref{Thm_cen} and Corollary~\ref{Coro_iso} imply that the Kullback-Leibler information number is a crucial factor that determines the performance of change detection algorithms. Specifically, Corollary~\ref{Coro_iso} shows that, for different sensors with different pairs of densities $f^i_1(x)$ and $f^i_0(x)$, the sensor associated with a density pair bearing a larger Kullback-Leibler information number asymptotically leads to a smaller PFA, under the same CADD performance. Compared with the isolated scheme, Theorem~\ref{Thm_cen} shows that, in the centralized scheme, the sum of Kullback-Leibler information numbers $\mathcal{D}$ is used to quantify the relation between PFA and CADD, which can be intuitively explained as follows. 

In the next section, we propose a distributed change detection scheme and analyze its performance. Due to the information propagation among sensors, we show that the distributed scheme will outperform the isolated one, and the outperforming is reflected by the averaged partial sum over individual Kullback-Leibler information numbers.

\section{Distributed Change Detection and Large Deviation Analysis}\label{Distributed_Change_Detection}
In this section, a random gossip based distributed change detection algorithm is first introduced. Then, we model the information propagation in this distributed scheme as a Markov process. Finally, two-layer large deviation analysis is presented to analyze the performance of the proposed distributed algorithm.

First, we interpret the network as a non-directed graph $\mathcal{G}=(\mathcal{V},\mathcal{E})$, where $\mathcal{V}$ is the set of nodes with $|\mathcal{V}|=N$ and $\mathcal{E}$ is the set of edges. If node $i$ is connected to node $j$, then we have that edge $(i,j)\in \mathcal{E}$. The connection in graph $\mathcal{G}$ is represented by the following $N \times N$ symmetric adjacency matrix $\mathcal{A}$ with each element $\mathcal{A}_{ij}$ as:
\begin{equation}
\mathcal{A}_{ij}=
\left\{ \begin{array}{l}
 1,~~ (i,j)\in \mathcal{E}~\mbox{or}~i=j, \\
 0,~~\mbox{otherwise}.
 \end{array} \right.
\end{equation}
We assume that the network is connected, i.e., each node has a path to any other node.

\subsection{Distributed Algorithm}
We propose a random gossip based distributed change detection algorithm, where a random gossip algorithm, as the inter-sensor communication structure, is used to propagate information among sensors within the neighborhood.

Communication among sensors is constrained by factors such as proximity, transmitting power, and receiving capabilities. We model the communication structure in terms of the non-directed graph $\mathcal{G}=(\mathcal{V},\mathcal{E})$, which is defined at the beginning of this section. If node $i$ can communicate with node $j$, there is an edge existing between $i$ and $j$, i.e., the set of edges $\mathcal{E}$ contains the edge $(i,j)$. We assume that the diagonal elements in adjacency matrix $\mathcal{A}$ are identically 1, which indicates that a node can always communicate with itself. The set $\mathcal{E}$ is the maximal set of allowable communication links in the network at any time; however, at a particular instant, only a fraction of the allowable communication links are active, for example, to avoid strong interference among communications. The exact communication protocol is not that important for the theoretical analysis, as long as the connectivity of network is satisfied. For definiteness, we assume the following generic communication model, which subsumes the widely used gossip protocol for real-time embedded architectures \cite{Boyd-Gossip} and the graph matching based communication protocols for internet architectures \cite{Mckeown}. Define the set $\mathcal{M}$ of binary symmetric $N \times N$ matrices as follows: \begin{equation}\label{setM}
\mathcal{M}=\left\{ \mathbf{A}| \mathbf{1}^T \mathbf{A}=\mathbf{1}^T,~\mathbf{A}\mathbf{1}=\mathbf{1},~\mathbf{A} \leq \mathcal{A}\right\}
\end{equation}
where $\mathbf{A}\leq \mathcal{A}$ is interpreted as component-wise. In other words, $\mathcal{M}$ is the set of adjacency matrices, where each node is incident to exactly one edge, which is included in the edge set $\mathcal{E}$. Let $\mathfrak{D}$ denote a probability distribution on the space $\mathcal{M}$. We define a sequence of time-varying matrices $\mathbf{A}(m)$, $m=1,2,\cdots$, as an independent and identically distributed sequence in $\mathcal{M}$ with distribution $\mathfrak{D}$. Define the averaged matrix $\bar{\mathbf{A}}$ as
\begin{equation}\label{bar_A}
 \bar{\mathbf{A}}=\int_{\mathcal{M}} \mathbf{A} d \mathfrak{D}(\mathbf{A}).
 \end{equation}

According to the definition of $\mathcal{M}$ in \eqref{setM}, $\bar{\mathbf{A}}$ is a symmetric stochastic matrix. We assume $\bar{\mathbf{A}}$ to be irreducible and aperiodic. This assumption depends on the allowable edges $\mathcal{E}$ and the distribution $\mathfrak{D}$. Such a distribution $\mathfrak{D}$ making this assumption valid always exists if the graph $(\mathcal{V},\mathcal{E})$ is connected, e.g., the uniform distribution. In addition, $\bar{\mathbf{A}}$ could be interpreted as the transition matrix of a Markov chain, which we will discuss later.

Assume that the sampling time interval for taking observations is $\Delta$, within which there are $M$ rounds of inter-sensor communications, where $M$ is a Poisson random variable with mean $\gamma$ \cite{Boyd-Gossip}. At the $m$-th ($m\in\{1,\cdots,M\}$) round, a node randomly selects another node from its neighborhood to construct a two-way communication pair to exchange the observations between them. At each sampling time interval, this communication structure is modeled by the sequence of matrices $\mathbf{A}(m)$, $m=1,2,\cdots,M$, i.e., the establishment of a communication link between node $i$ and node $j$ indicates that nodes $i$ and $j$ are neighbors with respect to the time varying adjacency matrix $\mathbf{A}(m)$. Note that there may exist multiple communication links or pairs simultaneously in the network, but only one communication link is associated with one given node in each round, which is also implied by the mathematical model in \eqref{setM}.

Now we model the communication link formation process from the perspective of Markov process. To this end, the communication link process governed by the time varying adjacency matrix sequence $\{\mathbf{A}(m)\}$ can be represented by $N$ particles traveling on the graph \cite{Di-LDKalman}. We denote the state of the $i$-th particle as $z_i(m)$, where $z_i(m)$ indicates the index of node that the $i$-th particle travels to at time $m$, with $z_i(m) \in \{1,\cdots,N\}$. The evolution of the $i$-th particle is given as follows:
\begin{equation}
z_i(m)=[z_i(m-1)]^{\rightarrow}_m,~z_i(0)=i,
\end{equation}
where the notation $[i]^{\rightarrow}_m$ denotes the neighbor of node $i$ at time $m$ with respect to the adjacency matrix $\mathbf{A}(m)$, i.e., a communication link is established between $i^{\rightarrow}_m$ and $i$ at time $m$. Thus, the travelling process of the $i$-th particle can be viewed as originating from node $i$ initially and then traveling on the graph according to the link formation process $\{\mathbf{A}(m)\}$ (possibly changing its location at each step). For each $i$, the process $\{z_i(m)\}$ is a Markov chain on the state space $\{1,\cdots,N\}$ with the transition probability matrix $\bar{\mathbf{A}}$ \cite{Di-LDKalman}.

After $M$ rounds of inter-sensor communications, each node accumulates some observations from other nodes, with which the local test statistic at each node is updated. Denote $O^i_n$ as the set of nodes whose observations are available at node $i$ after inter-sensor communications at the end of the observation time period $n$. We will describe the accumulation process to obtain $O^i_n$ later. Then, the distributed test statistic $\Lambda^i_{n,D}$ is updated as
\begin{equation}\label{test_stat_dist}
\Lambda^i_{n,D}=\frac{1}{1-\rho}(\Lambda^i_{n-1}+\rho)\prod_{j\in O_n^i}\frac{f^j_1(X^j_n)}{f^j_0(X^j_n)}.
\end{equation}

With this test statistic updating rule, at each sensor $i$, the distributed change detection scheme is executed with the following stopping time ${\tau_D^i}$:
\begin{equation}\label{opt_test_dist}
{\tau_D^i}(A)=\inf\{n\geq 1: \Lambda^i_{n,D}\geq A\},
\end{equation}
where $A$ is chosen as $A=(1-\alpha)/\alpha$ such that $\mbox{PFA}({\tau_D^i}(A)) \leq \alpha$.

Now we describe the observation accumulation process to obtain $O^i_n$.
Let $\mathbf{s}^m_n=[s^m_n(1),\cdots,s^m_n(N)]$, with element $s^m_n(i)\in \{1,\cdots,N\}$ indexing the observation $X_n^{s^m_n(i)}$ at sensor $i$ just after the $m$-th round of communication in the observation time period $n$. The initial state is $s^0_n(i)=i$ at each sensor $i$, which means that at the beginning of the time slot $n$, each sensor $i$ only has its own observation $X_n^i$. When the communication starts, by following the communication model $\mathbf{A}(m)$, the observations $\{X_n^i\}_{i\in\{1,\cdots,N\}}$ travel across the network in the following way:
\begin{equation}
\mathbf{s}^m_n=\mathbf{A}(m)\mathbf{s}^{m-1}_n,~1\leq m \leq M.
\end{equation}
During these $M$ rounds of inter-sensor communications until the end of the time period $n$, each sensor stores observations exchanged from other sensors. Then, at the end of the time period $n$, observations from other sensors are accumulated at sensor $i$, and the set of sensors whose observations are available at sensor $i$ is denoted by
\begin{equation}
O_n^i=\bigcup_{m=0}^{M}\{s^m_n(i)\}.
\end{equation}

This observation accumulation process terminates at the end of the time period $n$. Then, a similar observation accumulation process repeats during the time period $n+1$, which is independent of the previous process. Therefore, the sequence $\{O_n^i\}$, as the set denoting observation indices which are available at sensor $i$ at the end of the $n$-th period, is an i.i.d. process.

To better describe our work in the sequel, we introduce some notations here. Let $\Psi$ denote the power set of node indices $\{1,\cdots,N\}$, where elements of $\Psi$ are indexed by $\nu$, with $\nu \in \{0,1,\cdots, 2^N-1\}$. We use $\Psi_0$ to denote the null set and $\Psi_{2^N-1}$ to denote the whole set of node indices. For technical convenience, we interpret sensors in the set $\Psi_{\nu}$ indexed by $\nu$ to be arranged in an ascending order with $j_1$ denoting the first one and $j_{|\Psi_\nu|}$ denoting the last one, i.e., $\Psi_{\nu}=\{j_1,\cdots,j_{|\Psi_\nu|}\}$. Therefore, the set $O_n^i$, denoting nodes whose observations are available at node $i$ after the observation accumulation process, is a random variable taking values from $\Psi$. We denote the following probability as
\begin{equation}\label{q_def}
\mbox{Pr}(O_n^i=\Psi_{\nu})=q^i_n(\nu),~\nu \in \{0,1,\cdots, 2^N-1\}.
\end{equation}

\subsection{First-layer Large Deviation Analysis}

To perform large deviation analysis, we first need to interpret the stopping time ${\tau^i_D}(A)$ as a form of random walk crossing a threshold plus a nonlinear term \cite{Tart-Bayesian}. To this end, the stopping time ${\tau^i_D}(A)$ could be rewritten as:
\begin{equation}\label{random_walk}
\tau^i_D(A)=\inf\left\{n\geq 1: W_n(\rho)+{l}_n\ \geq \log(A/\rho)\right\},
\end{equation}
where $W_n(\rho)=Z_n+n|\log(1-\rho)|$ is a random walk with
\begin{align}\label{Z_n}
Z_n=\sum_{k=1}^n \sum_{j\in O_k^i} \log \frac{f_0^j(X^j_k)}{f_1^j(X^j_k)},
\end{align}
and
\begin{equation}\label{l_n}
l_n=\log \left\{1+\sum_{k=1}^{n-1} (1-\rho)^k \prod_{s=1}^k \prod_{j\in O_s^i} \frac{f_0^j(X^j_s)}{f_1^j(X^j_s)}\right\}.
\end{equation}
Specifically, $W_n(\rho)$ is a random walk with mean
\begin{equation}\label{W_n_mean}
\mathbb{E}_1\{W_n(\rho)\}=n\sum_{\nu=1}^{2^N-1}\bar{q}^i_\gamma(\nu)\sum_{j\in {\Psi}_\nu}D(f^j_1,f^j_0)+n|\log(1-\rho)|,
\end{equation}
where $\bar{q}^i_\gamma(\nu)$ is the probability defined as
\begin{equation}\label{q_def1}
\bar{q}^i_\gamma(\nu)=\mbox{Pr}(O_\gamma^i=\Psi_{\nu}),~\nu \in \{0,1,\cdots, 2^N-1\},
\end{equation}
in which $O_\gamma^i$, a random variable taking values from $\Psi$, denotes the set of nodes whose observations are available at node $i$ after $\gamma$ rounds of communications, and $\gamma$ is the mean value of the number of communication rounds. Then, based on the above random walk interpretation for the stopping time, we have the following theorem regarding the relation between PFA and CADD in the proposed distributed change detection scheme.
\begin{theorem}\label{thm_distributed}
The probability of false alarm ($\mbox{PFA}(\tau_D^i)$), with the stopping rule \eqref{opt_test_dist} in the distributed change detection algorithm with the parameter $\gamma$ as the averaged number of inter-sensor communications, satisfies the large deviation principle in the asymptotic sense with respect to increasing conditional ADD ($\mbox{CADD}_1(\tau_D^i)$), i.e.,
\begin{align}\label{eq_thm3}
\lim_{\mbox{CADD}_1(\tau_D^i)\rightarrow \infty} & \frac{1}{\mbox{CADD}_1(\tau_D^i)} \log [\mbox{PFA}(\tau_D^i)] \nonumber\\
&=-( \mathcal{D}^i_\gamma+|\log(1-\rho)|),
\end{align}
where $\mathcal{D}^i_\gamma=\sum_{\nu=1}^{2^N-1}\bar{q}^i_\gamma(\nu)\sum_{j\in {\Psi}_\nu}D(f^j_1,f^j_0)$, and $\mathcal{D}^i_\gamma+|\log(1-\rho)|$ is the large deviation decay rate of PFA. We call $\mathcal{D}^i_\gamma$ as the distributed Kullback-Leibler information number.
\end{theorem}

Theorem~\ref{thm_distributed} shows that $\mathcal{D}^i_\gamma$, whose function is similar to $\mathcal{D}$ in the centralized scheme and $D (f_1^i,f_0^i)$ in the isolated scheme, is a crucial factor determining the performance of the distributed change detection algorithm. The physical meaning of $\mathcal{D}^i_\gamma$ is explained as follows. Due to the observation propagation process, observations and the corresponding log-likelihood ratios from other sensors are available at each sensor; to some extent, $\mathcal{D}^i_\gamma$ can be considered as an accumulated form of these information. In particular, $\mathcal{D}^i_\gamma$ is an averaged partial sum of the Kullback-Leibler information numbers $D (f_1^i,f_0^i),~i=1,\cdots,N$, compared to $\mathcal{D}$ as the total sum. Also, from the mathematical form of $\mathcal{D}^i_\gamma$, we see that $ D (f_1^i,f_0^i) \leq \mathcal{D}^i_\gamma \leq \mathcal{D}$, and the case of $\bar{q}^i_\gamma(1)=1$ corresponds to the lower bound $\mathcal{D}^i_\gamma=D (f_1^i,f_0^i)$, while the case of $\bar{q}^i_\gamma(2^N-1)=1$ corresponds to the upper bound $\mathcal{D}^i_\gamma = \mathcal{D}$. Since $ D (f_1^i,f_0^i) \leq \mathcal{D}^i_\gamma \leq \mathcal{D}$ and $\mathcal{D}^i_\gamma$ determines the performance of the change detection algorithm, the above analysis proves that the distributed algorithm outperforms the isolated algorithm, but falls behind the centralized algorithm.

We present the proof for the above theorem as follows.
\begin{proof}
The proof adopts the relevant results from the nonlinear renewal theory in \cite{Woodroofe}.  
To complete the proof, we first present two preliminary results, regarding the proposed distributed algorithm, as follows:
\begin{align}
\mbox{PFA}(\tau_D^i)&=\frac{\zeta(\rho,\mathcal{D}^i_\gamma)}{A}(1+o(1)),~\mbox{as}~A \rightarrow \infty, \label{dis_thm5_1}
\end{align}
\begin{align}
\mathbb{E}_1(\tau_D^i)&=\frac{1}{\mathcal{D}^i_\gamma+|\log(1-\rho)|}\left[\log \frac{A}{\rho}-\xi(\rho,\mathcal{D}^i_\gamma)\right]+o(1),\nonumber\\
&~~~~~~~~~~~~~\mbox{as}~A \rightarrow \infty,\label{dis_thm5_2}
\end{align}
where $\mathcal{D}^i_\gamma$ is defined below \eqref{eq_thm3}, denoting the averaged value of the Kullback-Leibler information number in the distributed algorithm, and $\zeta(\rho,\mathcal{D}^i_\gamma)$ and $\xi(\rho,\mathcal{D}^i_\gamma)$
are functions of parameters $\rho$ and $\mathcal{D}^i_\gamma$.

Note that the above results for the distributed algorithm is similar to Theorem 5 in \cite{Tart-Bayesian}, which is related to the performance of the isolated algorithm. The difference is that the averaged partial sum of the Kullback-Leibler numbers is involved in the distributed algorithm, due to the observation accumulation at each node. In the sequel, we provide the proof flow for these two results. 

First, we verify \eqref{dis_thm5_1}. By recalling $p_n$ defined in \eqref{posterior_prob} and $\Lambda_n=p_n/(1-p_n)$, we have
\begin{align}\label{Thm_eq_E00}
\mbox{PFA}(\tau_D^i)&=\mathbb{E}^\pi(1-p_{\tau_D^i})=\mathbb{E}^\pi(1+\Lambda_{\tau_D^i})^{-1}\nonumber\\
&=\mathbb{E}^\pi\left(1+A\frac{\Lambda_{\tau_D^i}}{A}\right)^{-1}\nonumber\\
&=\frac{1}{A}\mathbb{E}^\pi \left(e^{-\omega_a}\right)(1+o(1)),~A\rightarrow \infty,
\end{align}
where $\omega_a=\log \Lambda_{\tau_D^i} -a$ and $a=\log(A/\rho)$. For $\mathbb{E}^\pi \left(e^{-\omega_a}\right)$, we have
\begin{align}\label{Thm_eq_E0}
\mathbb{E}^\pi \left(e^{-\omega_a}\right)&=\mathbb{E}^\pi \left(e^{-\omega_a}|{\tau_D^i\geq \lambda} \right) (1-\mbox{PFA}(\tau_D^i))\nonumber\\
&~~~+\mathbb{E}^\pi \left(e^{-\omega_a}{|\tau_D^i < \lambda}\right) \mbox{PFA}(\tau_D^i)\nonumber\\
&=\mathbb{E}^\pi \left( e^{-\omega_a}{|\tau_D^i\geq \lambda}\right)+O(A^{-1}),~A\rightarrow \infty,
\end{align}
which is due to $\mbox{PFA}(\tau_D^i) \leq 1/(1+A) < 1/A$.

Thus, we turn to study $\mathbb{E}^\pi \left( e^{-\omega_a}{|\tau_D^i\geq \lambda}\right)$ as
\begin{align}\label{thm_eq_E1}
&\mathbb{E}^\pi \left( e^{-\omega_a}{|\tau_D^i\geq \lambda}\right)\nonumber\\
&=\sum_{k=1}^{\infty} \mathbb{E}_k \left( e^{-\omega_a}{|\tau_D^i\geq k}\right) \mathbb{P}(\lambda=k|\tau_D^i\geq k).
\end{align}

For any $1\leq k <\infty$, we have
\begin{equation}
\tau^i_D=\inf\left\{n\geq 1: W_{n,k}(\rho)+{l}_{n,k} \geq a\right\},
\end{equation}
where $W_{n,k}(\rho)=Z_{n,k}+(n-k+1)|\log(1-\rho)|$,~$n\geq k$,~is a random walk with
$\mathbb{E}_k \left[W_{n,k}(\rho)\right]=(n-k+1)(\mathcal{D}^i_\gamma+|\log(1-\rho)|)$ and $l_{n,k}$ is a nonlinear term. In $W_{n,k}(\rho)$, we have
\begin{align}
Z_{n,k}=\sum_{t=k}^n \sum_{j\in O_t^i} \log \frac{f_0^j(X^j_t)}{f_1^j(X^j_t)}.
\end{align}

Then, by applying Theorem 4.1 in \cite{Woodroofe}, we obtain
\begin{equation}\label{thm_eq_E}
\lim_{A \rightarrow \infty} \mathbb{E}_k \left( e^{-\omega_a}{|\tau_D^i\geq k}\right) = \zeta(\rho,\mathcal{D}^i_\gamma),
\end{equation}
where $\zeta(\rho,\mathcal{D}^i_\gamma)$ is a function of parameters $\rho$ and $\mathcal{D}^i_\gamma$.

We also have
\begin{equation}\label{Thm_eq_p}
\lim_{A\rightarrow \infty} \mathbb{P}(\lambda=k|\tau_D^i\geq k)=\lim_{A\rightarrow \infty} \frac{\pi_k \mathbb{P}_k(\tau_D^i\geq k|\lambda=k)}{\mathbb{P}^\pi(\tau_D^i\geq k)}=\pi_k.
\end{equation}

Therefore, by plugging \eqref{thm_eq_E} and \eqref{Thm_eq_p} into \eqref{thm_eq_E1}, we have
\begin{equation}\label{Thm_eq_E2}
\lim_{A\rightarrow \infty} \mathbb{E}^\pi \left( e^{-\omega_a}{|\tau_D^i\geq \lambda}\right)=\zeta(\rho,\mathcal{D}^i_\gamma).
\end{equation}
Finally, by combining \eqref{Thm_eq_E00}, \eqref{Thm_eq_E0}, and \eqref{Thm_eq_E2}, we prove \eqref{dis_thm5_1}.

The proof of \eqref{dis_thm5_2} depends on Theorem 4.5 in \cite{Woodroofe}. In order to use this theorem, the validity of the following three conditions needs to be checked:
\begin{align}
&\sum_{n=1}^{\infty} \mathbb{P}_1 \{l_n \leq -\theta n\} < \infty,~\mbox{for~some}~0<\theta < \mathcal{D}_D^i;\nonumber\\
&\max_{0\leq k \leq n}|l_{n+k}|,~n\geq 1,~\mbox{are}~\mathbb{P}_1~\mbox{uniformly~integrable};\nonumber\\
&\lim_{A\rightarrow \infty}a\mathbb{P}_1\{\tau^i_D(A)\leq \varepsilon a (\mathcal{D}_D^i+|\log(1-\rho)|)^{-1}\}=0,\nonumber\\
&~~~~~~~~~~~~~\mbox{for~some}~0<\varepsilon<1,~\mbox{where}~a=\log(A/\rho),\nonumber
\end{align}
where $l_n$ is defined in \eqref{l_n}.

It is easy to check that the first condition is valid, as $l_n\geq0$. For the second condition, we have $\max_{0\leq k \leq n}|l_{n+k}|=l_{2n}$, since $l_n$, $n=1,2,\cdots$, are non-decreasing. Thus, to check that the second condition is valid, we only need to show that $l_n$, $n=1,2,\cdots$, are $\mathbb{P}_1$ uniformly integrable. To this end, we have that $l_n$ converges almost surely, as $n\rightarrow \infty$, to the following random variable
\begin{equation}
l=\log \left\{1+\sum_{k=1}^{\infty} (1-\rho)^k \prod_{s=1}^k \prod_{j\in O_s^i} \frac{f_0^j(X^j_s)}{f_1^j(X^j_s)}\right\}.
\end{equation}
By taking the expectation, we have
\begin{equation}
\mathbb{E}_1(l) \leq \log \left\{1+\sum_{k=1}^{\infty} (1-\rho)^k\right\}=\log \frac{1}{\rho}.
\end{equation}

Since $l_n$, $n=1,2,\cdots$, are non-decreasing, we have $l_n\leq l$. Then, we have $\mathbb{E}_1(l_n) < \infty$, implying the uniform integrability. Therefore, the second condition is satisfied.

Now we intend to show the validity of the third condition. 
According to Lemma 1 in \cite{Tart-Bayesian}, we have
\begin{equation}
\mathbb{P}_1\left\{\tau^i_D(A)\leq 1+(1-\epsilon) L_a\right\} \leq e^{-\phi_\epsilon a}+\beta(\epsilon,A),
\end{equation}
where $L_a=a (\mathcal{D}_D^i+|\log(1-\rho)|)^{-1}$, $\phi_\epsilon > 0$ for all $0 < \epsilon <1$, and $\beta(\epsilon,A)=\mathbb{P}_1\left\{ \max_{1\leq n < K_{\epsilon,A}} Z_n \geq (1+\epsilon) \mathcal{D}_D^i K_{\epsilon,A}\right\}$, in which $K_{\epsilon,A}=(1-\epsilon)L_a$ and $Z_n$ is defined in \eqref{Z_n}. The term $ e^{-\phi_\epsilon a}$ on the right-hand side is $o(1/a)$. 
Thus, in order to show
\begin{equation}
\lim_{A\rightarrow \infty}a \mathbb{P}_1\left\{\tau^i_D(A)\leq 1+(1-\epsilon) L_a\right\}=0,
\end{equation}
we only need to prove that the other term $\beta(\epsilon,A)$ is also $o(1/a)$, since $a=\log(A/\rho)$. To this end, by applying Theorem 1 of \cite{Chow}, for $\nu>0$ and $r\geq 0$, we have
\begin{align}
&\sum_{n=1}^{\infty}\mathbb{P}_1\left\{ \max_{1\leq k \leq n} (Z_k-\mathcal{D}_D^i k) \geq \nu n\right\} \nonumber\\
&\leq C_r\left\{ \mathbb{E}_1 [(Z_1-\mathcal{D}_D^i)^+]^{r+1}+[\mathbb{E}_1(Z_1-\mathcal{D}_D^i)^2]^r\right\},
\end{align}
where $C_r$ is a constant. When $r=1$, the finiteness of the right-hand side of the above inequality implies that the left-hand side is also finite. Thus, we obtain $\mathbb{P}_1\left\{ \max_{1\leq k \leq n} (Z_k-\mathcal{D}_D^i k) \geq \nu n\right\}=o(1/n)$.

Then, with the fact that
\begin{equation}
\beta(\epsilon,A) \leq \mathbb{P}_1\left\{ \max_{1\leq n < K_{\epsilon,A}} (Z_n-\mathcal{D}_D^i n) \geq \epsilon \mathcal{D}_D^i K_{\epsilon,A}\right\},
\end{equation}
we have $\beta(\epsilon,A)=o(1/a)$. Therefore,
\begin{equation}
\lim_{A\rightarrow \infty}a \mathbb{P}_1\left\{\tau^i_D(A)\leq 1+(1-\epsilon) L_a\right\}=0.
\end{equation}
By taking $\varepsilon =1-\epsilon$, finally we have
\begin{align}
&\lim_{A\rightarrow \infty}a\mathbb{P}_1\{\tau^i_D(A)\leq \varepsilon L_a\} \nonumber\\
&\leq \lim_{A\rightarrow \infty}a \mathbb{P}_1\left\{\tau^i_D(A)\leq 1+(1-\epsilon) L_a\right\}\nonumber\\
&=0
\end{align}
Hence the third condition is satisfied. Therefore, the conditions of Theorem 4.5 in \cite{Woodroofe} are satisfied. This theorem shows that \eqref{dis_thm5_2} is valid.

Then, with \eqref{dis_thm5_1} and \eqref{dis_thm5_2}, by taking the same proof method of Theorem~\ref{Thm_cen}, we have
\begin{align}
\lim_{\mbox{CADD}_1(\tau_D^i)\rightarrow \infty} & \frac{1}{\mbox{CADD}_1(\tau_D^i)} \log [\mbox{PFA}(\tau_D^i)]\nonumber\\
&=-( \mathcal{D}^i_\gamma+|\log(1-\rho)|).
\end{align}
\end{proof}

\subsection{Second-layer Large Deviation Analysis}
 Since $\mathcal{D}^i_\gamma$ has been shown as a crucial factor in the large deviation analysis of last subsection, in this subsection, we focus on studying the behavior of $\mathcal{D}^i_\gamma$. As we still stay in the scope of large deviation analysis as we did in the last subsection, we call it as the second-layer large deviation analysis, where the analysis in the last subsection is called the first-layer large deviation analysis.

As we cannot obtain the closed-form for $\mathcal{D}^i_\gamma$ due to the complicated probabilities incorporated, we discuss its asymptotic behavior when $\gamma \rightarrow \infty$. To this end, we first study the behavior of $\bar{q}^i_\gamma(\nu)$, defined below \eqref{W_n_mean}, when $\gamma \rightarrow \infty$, by employing the concept of hitting times in Markov chains.

For each $\nu \neq 2^N-1$, without loss of generality, we assume that $\nu$ corresponds to the index of the sensor subset $\{i_1,i_2,\cdots,i_m\}$, with $\{i'_1,i'_2,\cdots,i'_{N-m}\}$ as the complementary subset, where $m\geq 1$ due to the fact that at least its own observation is available at each sensor. Let $T_j$ denote the hitting time, starting from state (index of sensor) $j$ to hit another specific state $i$ in the Markov chain, whose transition probability matrix is $\bar{\mathbf{A}}$ defined in \eqref{bar_A}. From Theorem 7.26 in \cite{Driver-Markov}, since the transition probability matrix $\bar{\mathbf{A}}$ is irreducible, there exists constants $ 0< \alpha <1$ and $ 0 <L < \infty$ such that $P(T_{j}>L) \leq \alpha,\forall j$, and more generally,
\begin{equation} \label{Driver1}
  P(T_{j}>kL) \leq \alpha^k, ~k=0,1,2,\cdots.
\end{equation}
Also, there exists a constant $0<\beta <1$ such that $P(T_{j}>L) \geq \beta,\forall j$, and more generally,
\begin{equation}\label{Driver2}
  P(T_{j}>kL) \geq \beta^k, ~k=0,1,2,\cdots.
\end{equation}

Based on the above results of hitting times in Markov chains, we first present the following large deviation related theorem on the asymptotic behavior of $\sum_{v=0}^{2^N-2}\bar{q}^i_\gamma(v)$, as $\gamma \rightarrow \infty$. Since $\nu \in \{0,1,\cdots, 2^N-1\}$ according to \eqref{q_def1}, we have $\sum_{v=0}^{2^N-2}\bar{q}^i_\gamma(v)=1-\bar{q}^i_\gamma(2^N-1)$, where $\bar{q}^i_\gamma(2^N-1)$ denotes the probability that the observations from all sensors are available at sensor $i$, i.e., $\sum_{v=0}^{2^N-2}\bar{q}^i_\gamma(v)$ is the probability of the event that not all observations are available at sensor $i$.

\begin{theorem}\label{thm_second}
As $\gamma \rightarrow \infty$, the probability $\sum_{v=0}^{2^N-2}\bar{q}^i_\gamma(v)$ has the large deviation upper and lower bounds as follows,
\begin{equation}
 \frac{\ln \beta}{L}\leq \lim_{{\gamma} \rightarrow \infty} \frac{1}{{\gamma}} \ln \sum_{v=0}^{2^N-2}\bar{q}^i_\gamma(v) \leq \frac{\ln \, \alpha}{L},
\end{equation}
where $\alpha$, $\beta$ and $L$ are parameters in \eqref{Driver1} and \eqref{Driver2}.
\end{theorem}

Since $\sum_{v=0}^{2^N-2}\bar{q}^i_\gamma(v)$ presents the probability of the event that not all observations are available at sensor $i$, Theorem~\ref{thm_second} implies that this event is a rare event and its probability decays exponentially fast to zero as $\gamma \rightarrow \infty$.

The proof is presented as follows.
\begin{proof}
 Recall that $\nu$ corresponds to the index of the sensor subset $\{i_1,i_2,\cdots,i_m\}$, with $\{i'_1,i'_2,\cdots,i'_{N-m}\}$ as the complementary subset, and $T_j$ denotes the hitting time, starting from state (index of sensor) $j$ to hit another specific state $i$ in the Markov chain. Then, the probability $\bar{q}^i_\gamma(\nu)$ could be represented as
\begin{align}\label{q_up_bound}
\bar{q}^i_\gamma(\nu)&=\mbox{Pr}(T_{i'_1}>\gamma,\cdots,T_{i'_{N-m}}>\gamma,T_{i_1} \leq \gamma,\cdots,T_{i_m} \leq \gamma)\nonumber\\
&\leq \mbox{Pr}(T_{i'_1}>\gamma,\cdots,T_{i'_{N-m}}>\gamma) \nonumber\\
&\leq \min_{1\leq n \leq N-m} \mbox{Pr}(T_{i'_n}>\gamma).
\end{align}

Thus, we have
\begin{align}\label{upperbound_Oprotocol1}
\lim_{{\gamma} \rightarrow \infty} &\frac{1}{{\gamma}} \ln\left(\bar{q}^i_\gamma(\nu)\right) \leq \lim_{{\gamma} \rightarrow \infty} \frac{1}{{\gamma}} \ln \left(  \min_{1 \leq n \leq {N-m}}
P(T_{i'_n}>\gamma) \right) \nonumber \\
&\leq  \lim_{{\gamma} \rightarrow \infty} \frac{1}{{\gamma}} \ln \left( \alpha^{ \lfloor \gamma/L \rfloor} \right)= \frac{\ln \, \alpha}{L}
\end{align}
where the second inequality is due to \eqref{Driver1}.

For $\bar{q}^i_\gamma(\nu)$, we also have
\begin{align}\label{q_lower_bound}
\bar{q}^i_\gamma(\nu)&=\mbox{Pr}(T_{i'_1}>\gamma,\cdots,T_{i'_{N-m}}>\gamma,T_{i_1} \leq \gamma,\cdots,T_{i_m} \leq \gamma) \nonumber\\
&\geq \mbox{Pr}(T_{i'_1}>\gamma)\cdots \mbox{Pr}(T_{i'_{N-m}}>\gamma) \nonumber\\
&~~~~\mbox{Pr}(T_{i_1} \leq \gamma) \cdots \mbox{Pr}(T_{i_m} \leq \gamma).
\end{align}

This leads to
\begin{align}\label{lowerbound_Oprotocol1}
\lim_{{\gamma} \rightarrow \infty}& \frac{1}{{\gamma}} \ln\left(\bar{q}^i_\gamma(\nu)\right) \nonumber\\
& \geq \lim_{{\gamma} \rightarrow \infty} \frac{1}{{\gamma}} \ln\left[  \left(\beta^{ \lceil \gamma/{L} \rceil }\right)^{N-m}   \left(1- \alpha^{ \lfloor \gamma/{L} \rfloor }  \right)^m \right] \nonumber\\
& = (N-m) \frac{\ln \beta}{L}
\end{align}
where the first inequality is due to \eqref{Driver1} and \eqref{Driver2}, and the last equality is derived with $0 < \alpha <1$.

By combining \eqref{upperbound_Oprotocol1} and \eqref{lowerbound_Oprotocol1}, we have
\begin{equation}\label{up_lower_bound}
 (N-m) \frac{\ln \beta}{L} \leq \lim_{{\gamma} \rightarrow \infty} \frac{1}{{\gamma}} \ln\left(\bar{q}^i_\gamma(\nu)\right) \leq \frac{\ln \, \alpha}{L}.
\end{equation}


Then, we obtain
\begin{align}\label{up}
 &\lim_{{\gamma} \rightarrow \infty} \frac{1}{{\gamma}} \ln \sum_{v=0}^{2^N-2}\bar{q}^i_\gamma(v)\nonumber\\
 &\leq \lim_{{\gamma} \rightarrow \infty} \frac{1}{{\gamma}} \ln \left[(2^N-1) \max_v (\bar{q}^i_\gamma(v))\right]\nonumber\\
 &=\lim_{{\gamma} \rightarrow \infty} \frac{1}{{\gamma}} \ln \left[\max_v(\bar{q}^i_\gamma(v))\right]\nonumber\\
 &\leq \frac{\ln \, \alpha}{L},
\end{align}
where the last inequality is due to \eqref{up_lower_bound}.

We also have
\begin{align}\label{low}
 \lim_{{\gamma} \rightarrow \infty} \frac{1}{{\gamma}} \ln \sum_{v=0}^{2^N-2}\bar{q}^i_\gamma(v)&\overset{a}{\geq} \lim_{{\gamma} \rightarrow \infty} \frac{1}{{\gamma}} \ln  (\bar{q}^i_\gamma(v_p))\nonumber\\
 &\overset{b}{=}\frac{\ln \beta}{L},
\end{align}
where $v_p$ on the right-hand side of inequality $a$ denotes a particular index of the subset of sensors such that $m=N-1$, i.e., $v_p$ is the index of the sensor subset $\{i_1,i_2,\cdots,i_{N-1}\}$, recalling the notations defined at the beginning of this section. Since for $v_p \in \{0,2^N-2\}$, we have $\sum_{v=0}^{2^N-2}\bar{q}^i_\gamma(v) \geq \bar{q}^i_\gamma(v_p)$, implying the establishment of inequality $a$. According to \eqref{up_lower_bound} and taking $m=N-1$, we derive the equation $b$ in \eqref{low}.

By combining \eqref{up} and \eqref{low}, we conclude that
\begin{equation}
 \frac{\ln \beta}{L}\leq \lim_{{\gamma} \rightarrow \infty} \frac{1}{{\gamma}} \ln \sum_{v=0}^{2^N-2}\bar{q}^i_\gamma(v) \leq \frac{\ln \, \alpha}{L}
\end{equation}
\end{proof}

Based on Theorem~\ref{thm_second}, we further have the following theorem regarding the behavior of the distributed Kullback-Leibler information number $\mathcal{D}^i_\gamma$ defined in Theorem~\ref{thm_distributed}.
\begin{theorem} \label{thm_third}
As $\gamma \rightarrow \infty$, we have the following upper and lower bounds for the value of $\mathcal{D}^i_\gamma$,
\begin{align}
\mathcal{D}- & \left[\max_{j\in \{1,\cdots,N\}\setminus i} D(f^j_1,f^j_0)\right] e^{ \frac{\ln \, \alpha}{L}\gamma}  \leq \mathcal{D}^i_\gamma \nonumber\\
&\leq \mathcal{D} - \left[\min_{j\in \{1,\cdots,N\}\setminus i} D(f^j_1,f^j_0)\right] e^{\frac{\ln \beta}{L}\gamma },
\end{align}
where $D(f^j_1,f^j_0)$ is the Kullback-Leibler information number defined in \eqref{KL_is} and $\mathcal{D}$ is the centralized Kullback-Leibler information number defined in Theorem~\ref{Thm_cen}, and $\ln \alpha/L$ and $\ln \beta /L$ are the upper and lower bounds derived in Theorem~\ref{thm_second}.
\end{theorem}
Theorem~\ref{thm_third} implies that $\mathcal{D}^i_\gamma$ converges to $\mathcal{D}$ exponentially fast, as $\gamma \rightarrow \infty$. Since $\mathcal{D}^i_\gamma$ and $\mathcal{D}$ determine the performance of the distributed and centralized algorithms respectively, this theorem also implies that the performance of the proposed distributed algorithm converges to that of the centralized one at an exponentially fast rate.
\begin{proof}
Recall $\mathcal{D}^i_\gamma=\sum_{\nu=1}^{2^N-1}\bar{q}^i_\gamma(\nu)\sum_{j\in {\Psi}_\nu}D(f^j_1,f^j_0)$ and $\mathcal{D}=\sum_{i=1}^N D (f_1^i,f_0^i)$. We have
\begin{align}\label{thm_eq_D}
\mathcal{D}^i_\gamma& \overset{a}{=}\bar{q}^i_\gamma(2^N-1) \mathcal{D} + \sum_{\nu=1}^{2^N-2}\bar{q}^i_\gamma(\nu)\sum_{j\in {\Psi}_\nu}D(f^j_1,f^j_0) \nonumber\\
&\overset{b}{=}\left(1-\sum_{\nu=1}^{2^N-2}\bar{q}^i_\gamma(\nu)\right) \mathcal{D} +  \sum_{\nu=1}^{2^N-2}\bar{q}^i_\gamma(\nu)\sum_{j\in {\Psi}_\nu}D(f^j_1,f^j_0),
\end{align}
where equation $a$ is due to the fact that ${\Psi}_\nu=\{1,\cdots,N\}$ with $\nu = 2^N-1$, i.e., ${\Psi}_{2^N-1}$ denotes the set of indices of all sensors, and equation $b$ is based on $\sum_{\nu=1}^{2^N-1}\bar{q}^i_\gamma(\nu)=1$.

Then, from \eqref{thm_eq_D}, we have
\begin{align}\label{thm3_eq1}
\mathcal{D}^i_\gamma & \leq \left(1-\sum_{\nu=1}^{2^N-2}\bar{q}^i_\gamma(\nu)\right) \mathcal{D} \nonumber \\
&+ \sum_{\nu=1}^{2^N-2}\bar{q}^i_\gamma(\nu) \max_{1 \leq \nu \leq 2^N-2}\sum_{j\in {\Psi}_\nu}D(f^j_1,f^j_0) \nonumber\\
&= \mathcal{D} - \sum_{\nu=1}^{2^N-2}\bar{q}^i_\gamma(\nu) \min_{j\in \{1,\cdots,N\}\setminus i} D(f^j_1,f^j_0).
\end{align}
We could also obtain
\begin{align}\label{thm3_eq2}
\mathcal{D}^i_\gamma & \geq \left(1-\sum_{\nu=1}^{2^N-2}\bar{q}^i_\gamma(\nu)\right) \mathcal{D} \nonumber \\
&+ \sum_{\nu=1}^{2^N-2}\bar{q}^i_\gamma(\nu) \min_{1 \leq \nu \leq 2^N-2}\sum_{j\in {\Psi}_\nu}D(f^j_1,f^j_0) \nonumber\\
&= \mathcal{D} - \sum_{\nu=1}^{2^N-2}\bar{q}^i_\gamma(\nu) \max_{j\in \{1,\cdots,N\}\setminus i} D(f^j_1,f^j_0).
\end{align}

According to Theorem~\ref{thm_second}, as $\gamma \rightarrow \infty$, we have
\begin{equation}\label{thm3_eq3}
e^{ \frac{\ln \beta}{L}\gamma} \leq \sum_{v=0}^{2^N-2}\bar{q}^i_\gamma(v) \leq e^{ \frac{\ln \, \alpha}{L}\gamma}.
\end{equation}

Then, by combining \eqref{thm3_eq1}, \eqref{thm3_eq2} and \eqref{thm3_eq3}, as $\gamma \rightarrow \infty$, we derive
\begin{align}
\mathcal{D}- & \left[\max_{j\in \{1,\cdots,N\}\setminus i} D(f^j_1,f^j_0)\right] e^{ \frac{\ln \, \alpha}{L}\gamma}  \leq \mathcal{D}^i_\gamma \nonumber\\
&\leq \mathcal{D} - \left[\min_{j\in \{1,\cdots,N\}\setminus i} D(f^j_1,f^j_0)\right] e^{\frac{\ln \beta}{L}\gamma }
\end{align}
\end{proof}
\section{Simulation Results}\label{simulation}

In this section, we simulate the proposed distributed algorithm with a network of 5 nodes taking observations. We consider a Bayesian setup, and set the prior distribution of the change-point time as a geometric distribution with parameter $\rho=0.1$. Before the change happens, we consider that the observation at each node follows a Gaussian distribution $N(0,1)$; after the change happens, the observation at node $i,~i=1,\cdots,5$, turns to follow another Gaussian distribution $N(0.1 \times i,1)$. Note that here we consider a setup that observations at different nodes have different post-change distributions, which is to mimic the more general situation that different nodes could suffer different levels of impact from the same physical change. For example, certain physical event, such as the leakage of chemical gas or the abrupt increasing of temperature, would lead to different degrees of impacts in different nodes, due to their various locations. The nodes near the origin of the physical event could suffer from a more serious influence, which is reflected by a larger mean in the post-distribution; the nodes faraway the origin could suffer from a less serious influence, which is reflected by a smaller mean in the post-distribution.

In Fig.~\ref{dis_cen_is}, we show the simulated and analytical results corresponding to the first-layer large deviation analysis, and also compare the performance of the distributed scheme versus the centralized and isolated ones. In the simulation, we set $\gamma$ as 6, recalling that $\gamma$ is the mean value for number of communication rounds within each sampling time period. In Fig.~\ref{dis_cen_is}, the dashed curves denote the simulated decay rate, and the solid lines present the analytical decay rates in Theorem~\ref{Thm_cen} for the centralized scheme, Corollary~\ref{Coro_iso} for the isolated scheme, and Theorem~\ref{thm_distributed} for the distributed scheme, respectively. A higher decay rate implies a lower PFA under the same conditional ADD, which means that the performance is better. Therefore, from the simulation results of the decay rates, we see that the the distributed scheme outperforms the isolated one, but performs worse than the centralized one, which conforms to the analytical result from Theorem \ref{thm_distributed}.

In Fig.~\ref{upper_lower_bound}, we show the simulation results of the value $\frac{1}{{\gamma}} \ln \sum_{v=0}^{2^N-2}\bar{q}^i_\gamma(v)$, denoting the decay rate of the rare event that not all observations are available at sensor $i$, as the parameter $\gamma$ increases, which is the second-layer large deviation analysis shown in Theorem~\ref{thm_second}. We also present the large deviation lower and upper bounds in Fig.~\ref{upper_lower_bound}, from which we see that the simulated decay rate locates between the large deviation lower and upper bounds, and the bounds are relatively tight, which verifies the analytical result in Theorem~\ref{thm_second}. Here we also present the lower bound ${\ln \beta}/{L}$ and the upper bound ${\ln \alpha}/{L}$ of Theorem~\ref{thm_second}, which are shown in Fig.~\ref{upper_lower_bound}. Recall that $T_j$ denotes the hitting time, starting from state (index of sensor) $j$ to hit another specific state $i$ in the Markov chain with the transition probability matrix $\bar{\mathbf{A}}$. Then we have
\begin{equation}
P(T_j > L)=\sum_{i_1,\cdots,i_L \neq i} \bar{\mathbf{A}}_{ji_1} \bar{\mathbf{A}}_{i_1i_2}\bar{\mathbf{A}}_{i_2i_3}\cdots\bar{\mathbf{A}}_{i_{L-1}i_L}.
\end{equation}
Recall that we intend to find $\alpha$ such that $P(T_{j}>L) \leq \alpha,~\forall j$. Thus, we can set $\alpha=\max_j P(T_j > L)$. In order to find $\beta$ such that $P(T_{j}>L) \geq \beta,~\forall j$, we can set $\beta=\min_j P(T_j>L)$. Then, we are ready to calculate ${\ln \alpha}/{L}$ and ${\ln \beta}/{L}$. To this end, the selection of $L$ is a critical step, as both $\alpha$ and $\beta$ are calculated based on the selection of $L$. Here we show the calculation of ${\ln \alpha}/{L}$ and ${\ln \beta}/{L}$ with different $L$ values in Fig.~\ref{upper_lower_bound_calculation}. From Fig.~\ref{upper_lower_bound_calculation}, we see a very interesting phenomenon that these two bounds look converging as $L$ increases, although here we will not provide the mathematical proof of this result. This observation could imply some potential properties for hitting time in Markov chains. The further exploration with analytical analysis based on this observation will be left for our future work. Note that the upper and lower bounds in Fig.~\ref{upper_lower_bound} are set as the values calculated with $L=15$.

In Fig.~\ref{upper_lower_bound_D}, we show the simulation results for the distributed Kullback-Leibler information $\mathcal{D}_\gamma^i$, the value of the centralized Kullback-Leibler information $\mathcal{D}$, and the calculation results for the upper and lower bounds presented in Theorem~\ref{thm_third}. From Fig.~\ref{upper_lower_bound_D}, we see that the upper bound is a very tight bound, while the lower bound is relatively looser. However, the range of y-axis in this figure is very small from 0.3765 to 0.3810; so both the lower and upper bounds are tight bounds in this sense. We also see that the distributed Kullback-Leibler information $\mathcal{D}_\gamma^i$ converges to the centralized Kullback-Leibler information $\mathcal{D}$, as $\gamma$ increases, which implies that the performance of the distributed change detection scheme converges to that of the centralized one, since $\mathcal{D}^i_\gamma$ and $\mathcal{D}$ determine the performance of the distributed and centralized schemes, respectively.

\begin{figure}
  \centering
  \includegraphics[width=0.45\textwidth]{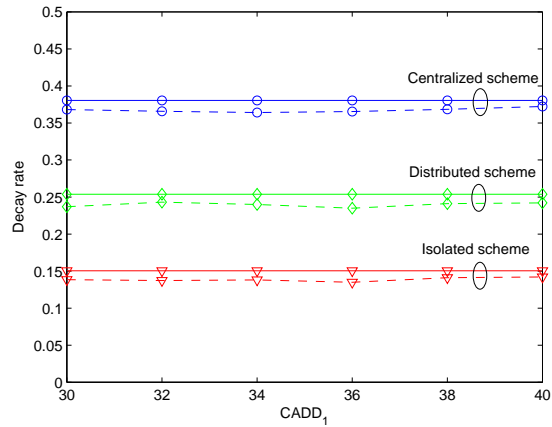}\\
  \caption{First-layer large deviation analysis: comparison of decay rates in distributed, centralized, and isolated schemes with simulation (dash curve) vs. analytical results (solid line).}\label{dis_cen_is}
\end{figure}

\begin{figure}
  \centering
  \includegraphics[width=0.45\textwidth]{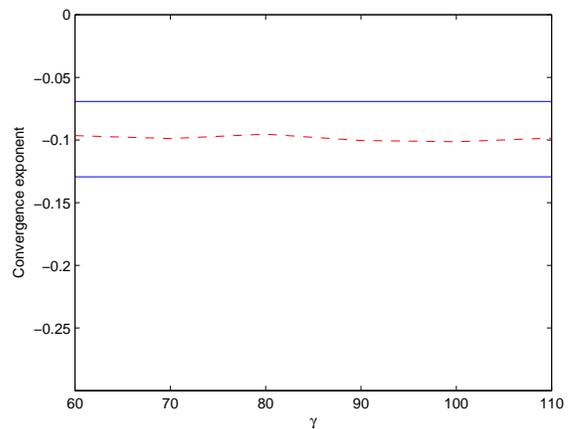}\\
  \caption{Second-layer large deviation analysis in Theorem~\ref{thm_second}: simulated decay rate (dash curve) of the probability of the rare event that not all observations are available at a sensor, and the corresponding large deviation upper and lower bounds (solid lines)}\label{upper_lower_bound}
\end{figure}

\begin{figure}
  \centering
  \includegraphics[width=0.45\textwidth]{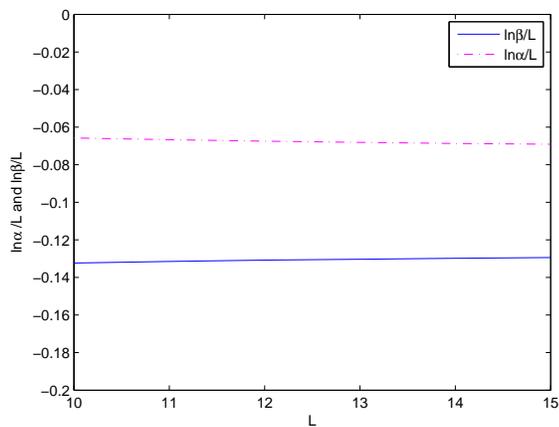}\\
  \caption{Calculation of the lower bound $\ln \beta /L$ and the upper bound $\ln \alpha /L$ in Theorem~\ref{thm_second} with varying $L$.}\label{upper_lower_bound_calculation}
\end{figure}

\begin{figure}
  \centering
  \includegraphics[width=0.45\textwidth]{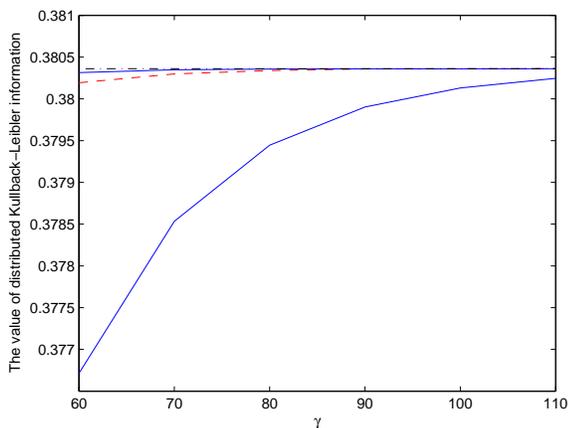}\\
  \caption{Simulated distributed Kullback-Leibler information (dash curve), centralized Kullback-Leibler information (dash-dot line) and the corresponding analytical upper and lower bounds (solid curve) in Theorem~\ref{thm_third}.}\label{upper_lower_bound_D}
\end{figure}

\section{Conclusion}\label{conclusion}
We have proposed a distributed Bayesian change detection scheme, with a random gossip-type protocol to realize inter-sensor communications. With this communication structure, we modeled the information propagation procedure in the network as a Markov process. We analyzed the performance of the proposed scheme via a method of two-layer large deviation analysis. The first-layer analysis proves that the probability of false alarm decays to zero at an exponentially fast rate, as the conditional averaged detection delay increases, and also shows that the Kullback-Leibler information number is a key factor determining the performance of the change detection algorithm. The second-layer analysis proves that the probability of the rare event that not all observations are available at a sensor decays exponentially fast to zero, as the averaged number of communications increases, where the large deviation upper and lower bounds of this decay rate are also derived. Then, we eventually prove that the performance of the distributed algorithm converges exponentially fast to that of the centralized one, by showing that the distributed Kullback-Leibler information number in the distributed algorithm converges to that in the centralized one.

\bibliographystyle{IEEEtran}
\bibliography{Bib_DCD,CentralBib}

\begin{thebibliography}{10}
\providecommand{\url}[1]{#1}
\csname url@samestyle\endcsname
\providecommand{\newblock}{\relax}
\providecommand{\bibinfo}[2]{#2}
\providecommand{\BIBentrySTDinterwordspacing}{\spaceskip=0pt\relax}
\providecommand{\BIBentryALTinterwordstretchfactor}{4}
\providecommand{\BIBentryALTinterwordspacing}{\spaceskip=\fontdimen2\font plus
\BIBentryALTinterwordstretchfactor\fontdimen3\font minus
  \fontdimen4\font\relax}
\providecommand{\BIBforeignlanguage}[2]{{%
\expandafter\ifx\csname l@#1\endcsname\relax
\typeout{** WARNING: IEEEtran.bst: No hyphenation pattern has been}%
\typeout{** loaded for the language `#1'. Using the pattern for}%
\typeout{** the default language instead.}%
\else
\language=\csname l@#1\endcsname
\fi
#2}}
\providecommand{\BIBdecl}{\relax}
\BIBdecl

\bibitem{Lai-QD-cognitive}
L.~Lai, Y.~Fan, and H.~Poor, ``Quickest detection in cognitive radio: A
  sequential change detection framework,'' in \emph{2008 IEEE Global
  Telecommunications Conference (GLOBECOM)}, New Orleans, USA, Nov. 2008, pp.
  1--5.

\bibitem{Husheng-Quickest}
H.~Li, H.~Dai, and C.~Li, ``Collaborative quickest spectrum sensing via random
  broadcast in cognitive radio systems,'' \emph{IEEE Transactions on Wireless
  Communications}, vol.~9, no.~7, pp. 2338--2348, July 2010.

\bibitem{Trivedi-secret}
S.~Trivedi and R.~Chandramouli, ``Secret key estimation in sequential
  steganography,'' \emph{IEEE Transactions on Signal Processing}, vol.~53,
  no.~2, pp. 746--757, Feb. 2005.

\bibitem{Thottan-Anomaly}
M.~Thottan and C.~Ji, ``Anomaly detection in {IP} networks,'' \emph{IEEE
  Transactions on Signal Processing}, vol.~51, no.~8, pp. 2191--2204, Aug.
  2003.

\bibitem{Tar-intrusion}
A.~Tartakovsky, B.~Rozovskii, R.~Blazek, and H.~Kim, ``A novel approach to
  detection of intrusions in computer networks via adaptive sequential and
  batch-sequential change-point detection methods,'' \emph{IEEE Transactions on
  Signal Processing}, vol.~54, no.~9, pp. 3372--3382, Sep. 2006.

\bibitem{Cardenas-Mac}
A.~Cardenas, S.~Radosavac, and J.~Baras, ``Evaluation of detection algorithms
  for mac layer misbehavior: Theory and experiments,'' \emph{IEEE/ACM
  Transactions on Networking}, vol.~17, no.~2, pp. 605--617, Apr. 2009.

\bibitem{Commenges-neuro}
D.~Commenges, J.~Seal, and F.~Pinatel, ``Inference about a change point in
  experimental neurophysiology,'' \emph{Mathematical Biosciences}, vol.~80,
  no.~1, pp. 81--108, July 1986.

\bibitem{Frisen-public}
M.~Frisen, ``Optimal sequential surveillance for finance, public health, and
  other areas,'' \emph{Sequential Analysis}, vol.~28, no.~3, pp. 310--337, July
  2009.

\bibitem{Sonesson-public}
C.~Sonesson and D.~Bock, ``A review and discussion of prospective statistical
  surveillance in public health,'' \emph{Journal of the Royal Statistical
  Society}, vol. 166, no.~1, pp. 5--21, July 2003.

\bibitem{Rice-structure}
J.~A. Rice, K.~Mechitov, S.~Sim, T.~Nagayama, S.~Jang, R.~Kim, B.~F. Spencer,
  G.~Agha, and Y.~Fujino, ``Flexible smart sensor framework for autonomous
  structural health monitoring,'' \emph{Smart Structures and Systems}, vol.~6,
  no.~5, pp. 423--438, May 2010.

\bibitem{Mainwaring}
A.~Mainwaring, D.~Culler, J.~Polastre, R.~Szewczyk, and J.~Anderson, ``Wireless
  sensor networks for habitat monitoring,'' in \emph{1st ACM International
  Workshop on Wireless Sensor Networks and Applications}, Atlanta, USA, Sep.
  2002, pp. 88--97.

\bibitem{VVV-EnergyEfficient}
T.~Banerjee and V.~Veeravalli, ``Energy-efficient quickest change detection in
  sensor networks,'' in \emph{2012 IEEE Statistical Signal Processing Workshop
  (SSP)}, Ann Arbor, USA, Aug. 2012, pp. 636--639.

\bibitem{Mei}
Y.~Mei, ``Quickest detection in censoring sensor networks,'' in \emph{2011 IEEE
  International Symposium on Information Theory Proceedings (ISIT)}, St.
  Petersburg, Russia, July 2011, pp. 2148--2152.

\bibitem{Tartakovsky-quickest03}
A.~G. Tartakovsky and V.~V. Veeravalli, ``Quickest change detection in
  distributed sensor systems,'' in \emph{6th IEEE International Conference on
  Information Fusion}, Cairns, Australia, July 2003, pp. 756--763.

\bibitem{Tartakovsky08asymptoticallyoptimal}
------, ``Asymptotically optimal quickest change detection in distributed
  sensor systems,'' \emph{Sequential Analysis}, vol.~27, no.~4, pp. 441--475,
  Oct. 2008.

\bibitem{VVV-decentr2001}
V.~Veeravalli, ``Decentralized quickest change detection,'' \emph{IEEE
  Transactions on Information Theory}, vol.~47, no.~4, pp. 1657--1665, May
  2001.

\bibitem{Poor-oneshot}
O.~Hadjiliadis, H.~Zhang, and H.~Poor, ``One shot schemes for decentralized
  quickest change detection,'' \emph{IEEE Transactions on Information Theory},
  vol.~55, no.~7, pp. 3346--3359, July 2009.

\bibitem{Mous-decentralized}
G.~Moustakides, ``Decentralized {CUSUM} change detection,'' in \emph{9th
  International Conference on Information Fusion}, Florence, Italy, July 2006,
  pp. 1--6.

\bibitem{LZ-decentra}
L.~Zacharias and R.~Sundaresan, ``Decentralized sequential change detection
  using physical layer fusion,'' \emph{IEEE Transactions on Wireless
  Communications}, vol.~7, no.~12, pp. 4999--5008, Dec. 2008.

\bibitem{Di-GlobalSIP}
D.~Li, L.~Lai, and S.~Cui, ``Quickest change detection and identification
  across a sensor array,'' in \emph{2013 IEEE Global Conference on Signal and
  Information Processing (GlobalSIP)}, Austin,~USA, Dec. 2013, pp. 145--148.

\bibitem{Ban-Efficiency-quickest}
T.~Banerjee, V.~Sharma, V.~Kavitha, and A.~JayaPrakasam, ``Generalized analysis
  of a distributed energy efficient algorithm for change detection,''
  \emph{IEEE Transactions on Wireless Communications}, vol.~10, no.~1, pp.
  91--101, Jan. 2011.

\bibitem{Braca-distributedchangeconsensus}
P.~Braca, S.~Marano, V.~Matta, and P.~Willett, ``Consensus-based {P}age's test
  in sensor networks,'' \emph{Signal Processing}, vol.~91, pp. 919--930, Apr.
  2011.

\bibitem{stankovic-distributedchange}
S.~S. Stankovic, N.~Ilic, M.~S. Stankovic, and K.~H. Johansson, ``Distributed
  change detection based on a consensus algorithm,'' \emph{IEEE Transactions on
  Signal Processing}, vol.~59, no.~12, pp. 5686--5697, Dec. 2011.

\bibitem{Di-LDKalman}
D.~Li, S.~Kar, J.~Moura, H.~Poor, and S.~Cui, ``Distributed {K}alman filtering
  over massive data sets: Analysis through large deviations of random riccati
  equations,'' \emph{IEEE Transactions on Information Theory}, vol.~61, no.~3,
  pp. 1351--1372, Mar. 2015.

\bibitem{Di-KalmanQuanize}
D.~Li, S.~Kar, F.~Alsaadi, A.~Dobaie, and S.~Cui, ``Distributed {K}alman
  filtering with quantized sensing state,'' \emph{IEEE Transactions on Signal
  Processing}, vol.~63, no.~19, pp. 5180--5193, Oct. 2015.

\bibitem{DEMBO-LD}
A.~Dembo and O.~Zeitouni, \emph{Large Deviations Techniques and
  Applications}.\hskip 1em plus 0.5em minus 0.4em\relax New York: Springer,
  1998.

\bibitem{Bucklew-book}
J.~Bucklew, \emph{Large {D}eviation {T}echniques in {D}ecision, {S}imulation,
  and {E}stimation}.\hskip 1em plus 0.5em minus 0.4em\relax New York: Wiley,
  1990.

\bibitem{Baj-LD}
D.~Bajovic, D.~Jakovetic, J.~Xavier, B.~Sinopoli, and J.~Moura, ``Distributed
  detection via gaussian running consensus: Large deviations asymptotic
  analysis,'' \emph{IEEE Transactions on Signal Processing}, vol.~59, no.~9,
  pp. 4381--4396, Sep. 2011.

\bibitem{Jak-Detection-LD}
D.~Jakoveti\'c, J.~Moura, and J.~Xavier, ``Distributed detection over noisy
  networks: Large deviations analysis,'' \emph{IEEE Transactions on Signal
  Processing}, vol.~60, no.~8, pp. 4306--4320, Aug. 2012.

\bibitem{Sahu-SPRT-TSP-15}
\BIBentryALTinterwordspacing
A.~K. Sahu and S.~Kar, ``Distributed sequential detection for gaussian binary
  hypothesis testing,'' \emph{IEEE Transactions on Signal Processing}, to
  appear. [Online]. Available: \url{http://arxiv.org/pdf/1411.7716v2.pdf}
\BIBentrySTDinterwordspacing

\bibitem{Poor}
H.~V. Poor and O.~Hadjiliadis, \emph{Quickest Detection}.\hskip 1em plus 0.5em
  minus 0.4em\relax Cambridge,~UK: Cambridge University Press, 2008.

\bibitem{Tart-Bayesian}
A.~G. Tartakovsky and V.~V. Veeravalli, ``General asymptotic {B}ayesian theory
  of quickest change detection,'' \emph{Theory of Probability and Its
  Applications}, vol.~49, no.~3, pp. 458--497, 2005.

\bibitem{Shiryaev-1963}
A.~N. Shiryaev, ``On optimum methods in quickest detection problems,''
  \emph{Theory of Probability and Its Applications}, vol.~8, no.~1, pp. 22--46,
  1963.

\bibitem{Shiryaev-1978}
------, \emph{Optimal Stopping Rules}.\hskip 1em plus 0.5em minus 0.4em\relax
  New York: Springer-Verlag, 1978.

\bibitem{Boyd-Gossip}
S.~Boyd, A.~Ghosh, B.~Prabhakar, and D.~Shah, ``Randomized gossip algorithms,''
  \emph{IEEE Transactions on Information Theory}, vol.~52, no.~6, pp.
  2508--2530, June 2006.

\bibitem{Mckeown}
N.~McKeown, A.~Mekkittikul, V.~Anantharam, and J.~Walrand, ``Achieving 100~\%
  throughput in an input-queued switch,'' \emph{IEEE Transactions on
  Communications}, vol.~47, no.~8, pp. 1260--1267, Aug. 1999.

\bibitem{Woodroofe}
M.~Woodroofe, \emph{Nonlinear Renewable Theory in Sequential Analysis}.\hskip
  1em plus 0.5em minus 0.4em\relax Philadelphia, PA: SIAM, 1982.

\bibitem{Chow}
Y.~S. Chow and T.~L. Lai, ``Some one-sided theorems on the tail distribution of
  sample sums with applications to the last time and largest excess of boundary
  crossings,'' \emph{Transactions of the American Mathematical Society}, vol.
  208, pp. 51--72, July 1975.

\bibitem{Driver-Markov}
B.~K. Driver, ``Introduction to stochastic processes {II},'' [Online].
  Available:~\url{http://www.math.ucsd.edu/~bdriver/math180C_S2011/Lecture%20Notes/180Lec6b.pdf}.

\end{thebibliography}
\end{document}